\documentclass[journal]{IEEEtran}

\ifCLASSINFOpdf
\else
\fi

\usepackage{amsmath}
\usepackage{graphicx}
\usepackage{subfigure}
\usepackage{amsmath}
\usepackage[linesnumbered, ruled, vlined]{algorithm2e}
\usepackage{algpseudocode}
\usepackage{cite}
\usepackage{multicol}
\usepackage{multirow}
\usepackage{mathtools}
\usepackage{amssymb}
\usepackage{color}
\usepackage{xcolor}

\usepackage[normalem]{ulem}
\usepackage{bm}
\usepackage{indentfirst}
\usepackage{amsthm}
\usepackage{nomencl}
\usepackage{booktabs}
\usepackage{cite}
\usepackage{siunitx}
\usepackage{marvosym}
\usepackage[numbers,sort&compress]{natbib}
\usepackage[font=small,skip=0pt]{caption}
\usepackage{colortbl}        
\definecolor{mygray}{gray}{.8}
\definecolor{mypink}{rgb}{.99,.91,.95}
\definecolor{mycyan}{cmyk}{.3,0,0,0}
\usepackage{wrapfig}
\usepackage{mathtools}  
\usepackage[skip=0.5\baselineskip]{caption}
\let\sss= \scriptscriptstyle  

\usepackage[hyphens]{url}
\usepackage[hidelinks]{hyperref}

\newtheorem{theorem}{Theorem}

\newtheorem{remark}{Remark}

\begin{document}

\title{Selling Renewable Utilization  Service to  Consumers via Cloud Energy Storage}

\author{Yu~Yang,~\IEEEmembership{Student Member,~IEEE,}
	Utkarsha Agwan,~\IEEEmembership{Student Member,~IEEE,}\\
	Guoqiang~Hu,~\IEEEmembership{Senior Member,~IEEE,}
	and~Costas~J.~Spanos,~\IEEEmembership{Fellow,~IEEE}
	\thanks{This  work  was  supported  by  the  Republic  of  Singapore’s  National  Research  Foundation  through  a  grant  to  the  Berkeley  Education  Alliance  for  Research  in  Singapore
		(BEARS)  for  the  Singapore-Berkeley  Building  Efficiency  and  Sustainability  in  the
		Tropics  (SinBerBEST)  Program.  BEARS  has  been  established  by  the  University  of  California,  Berkeley  as  a  center  for  intellectual  excellence  in  research  and  education  in
		Singapore.}
	\thanks{Yu Yang is with SinBerBEST, Berkeley Education 	Alliance for Research in Singapore, Singapore 138602 e-mail: (yu.yang@bears-berkeley.sg).}
	\thanks{Guoqiang Hu is with the School 	of Electrical and Electronic Engineering, Nanyang Technological University,
		Singapore, 639798 e-mail: (gqhu@ntu.edu.sg).}
	\thanks{Utkarsha Agwan and Costas J. Spanos are with the Department of Electrical Engineering and 	Computer Sciences, University of California, Berkeley, CA, 94720 USA email: (\{uagwan, spanos\}@berkeley.edu).}
}


\maketitle
\begin{abstract}
This paper proposes a cloud energy storage (CES) model for enabling greater utilization of local renewable generation by building consumers (BCs). As opposed to  most existing  ES sharing models that the energy storage operator (ESO) leases energy or power capacity to its customers, our CES model suggests the ESO to sell  renewable utilization service (RUS) for higher profitability.  Particularly, the customers request CES service in  the amounts of total local renewable generation they want to shift  to supply their demand  over the contracted time period. 
We propose a quadratic price model for the ESO charging  its customers by the requested RUS,  and formulate their interactions as a Stackelberg game which admits  an \emph{equilibrium}. 
We prove the CES model outperforms individual ES (IES) model in \emph{social welfare}. 
We demonstrate the CES model can provide 2-4 times profit to the ESO and bring higher economic benefits (i.e., cost reduction) to its customers over the existing ES sharing models. 
Moreover, we show the CES model can achieve near \emph{social optima} and high ES efficiency (i.e., ES utilization) which are not provided by the other ES models.  Particularly, this paper can work as an example how market design and sharing economy can shape the efficiency of energy systems. 
\end{abstract}


%
\IEEEpeerreviewmaketitle

\section{Introduction}
\IEEEPARstart{R}{enewable} energy  is crucial for transitioning to a sustainable and low-carbon energy system  \cite{weitemeyer2015integration}.  The continuing drop in renewable cost and increase in policy incentives and mandatory targets have boosted  renewable installation (e.g., wind and solar power, etc.) around the globe  \cite{GlobalRenewableEnergy, GlobalWindSolar}.  Whereas blending the volatile and non-dispatchable  renewable supply into consumer demand at scale  still remains to be resolved  due to their asynchronous pace. 
Energy storage (ES)  deems one effective solution  to offset  the temporal imbalance  \cite{stenclik2017maintaining, goebel2015bringing, arbabzadeh2019role}. However, ES is still expensive and thus  paring renewable generators with isolated ES will make the shaped renewable supply cost-prohibitive and raise the levelized cost above high-carbon alternatives  \cite{eia2016levelized, ziegler2019storage}.
At the same time, despite being capital-intensive, individual or isolated ES installations for renewable integration are usually under-utilized due to the volatility of renewable generation. This poses the possibility of designing appropriate ES business models to enable cost-effective renewable integration.

 \emph{Sharing economy}  has manifested  in transportation and housing systems \cite{barron2018sharing, agatz2012optimization},  suggesting its  potential to bring in new technologies into energy systems such as ES \cite{lombardi2017sharing}.  
 For example, multiple consumers can share their under-utilized private ES with each other, or jointly invest in a central  ES \cite{chakraborty2018sharing}. Besides, third-party providers can invest in ES  to provide storage service to their customers \cite{zhao2019virtual}. The underlying idea is to increase ES utilization and  \emph{marginal value} \footnote{The economic gain of per-unit ES investment.},  thus lowering the economic barriers for ES deployment.
Different sharing paradigms can enable different ES business models and applications. Third-party based ES sharing models run by an energy storage operator (ESO) are foreseeable  in future energy systems from the profitability and flexibility of providing  service at scale.  Moreover, ES consumers can overcome the high ES capital cost barriers and  embrace greater local renewable utilization.  Though the potential economic benefits of ES sharing are clear, achieving the objective is a non-trivial task, and depends on a well-designed ES business model to address the following challenges.
\begin{itemize}

\item  The ESO and consumers are both self-interested,  thus the ES sharing model must admit an \emph{equilibrium} that ensures both ESO profitability and  consumer incentives (i.e., higher economic benefits over individual ES installations). 

\item  As ES is capital-intensive,   the operation of the shared ES should be well coordinated so as to improve the ES utilization and economic benefits for the ESO.   However, the ES  consumers  are separate  and competitive. 

\item  Due to the diversity of demand and renewable generation,  the consumers generally require to charge and discharge the ES in heterogeneous multi-step patterns, which makes it difficult for the ESO  to uniformly price the ES service. 
\end{itemize}

Mainly due to the above challenges, most of the existing third-party based ES sharing models fail in securing  desirable ESO  profit or consumer incentives (see \cite{liu2017decision, liu2019research, zhao2017pricing, zhao2019virtual, he2019optimal}). Typically,  \cite{liu2017decision} proposed a cloud energy storage (CES) model for customers to harvest grid price arbitrage. 
However, the numeric results report  the maximal ESO profit rate  is only 4.6\%   (relative to the ES capital cost) even all consumers commit to  paying  a service fee equal to  individual ES installation (in such setting,  no incentive is actually provided to the consumers).
\cite{he2019optimal} also proposed a CES model and addressed the specific problem of ES service pricing. In their model, the ESO charges the consumers by their occupied energy capacity (in \si{\kilo\watt\hour}) and power capacity (in \si{\kilo\watt}) with a linear price model. 
However, we can infer from the numeric results that the ESO profit is  mainly   from  the economics of scale rather than increasing ES utilization through sharing (the ESO purchases ES  at scale with a price 40\% lower than individual purchase).
Besides, \cite{zhao2017pricing, zhao2019virtual} proposed a novel third-party based ES sharing model characterized by a two-stage structure: the ESO computes the optimal price to maximize its profit and the customers determine the optimal energy capacity to rent.  They also used  a  similar linear price model by the energy capacity (in \si{\kilo\watt\hour}).  
Notably, a systematic algorithm was proposed to search for the \emph{equilibrium} price by exploring the problem structures. However, the algorithm seems cumbersome for application. 
 Moreover, the profitability of ESO seems sensitive to scenarios as it depends on the  operation of the  self-interested and independent consumers  over   their contracted energy capacity.

 In a nutshell, most third-party based ES sharing models can not ensure ESO profitability and consumer incentives.  This is mainly caused by the inappropriate (linear) price model used for bargaining.
The objective to benefit both the ESO and consumers
 can be achieved if and only if the ES utilization is enhanced  through coordination and sharing.
  Whereas a linear price model based on  energy or power capacity can not direct the coordination of  customers over their   multi-step  charging and discharging of the shared ES. In such setting,  the ESO and consumers  are likely to  fail in bargaining and sharing the potential economic benefits of ES.     
More importantly, most of such models do not retain \emph{social welfare}, taking the \emph{sharing economy} away  from its pathway of  improving  resource efficiency  and creating  a sustainable future  \cite{mi2019sharing}.

\subsection{Contributions}

Motivated by the literature, especially~\cite{zhao2017pricing, zhao2019virtual},  this paper studies a third-party based  cloud energy storage (CES) model.  We seek to  achieve   ESO profitability  and  consumer incentives while retaining  \emph{social welfare}.   
In contrast to most existing works where the ESO profits by leasing energy or power capacity, our CES model suggests the ESO to sell renewable utilization service (RUS) to its consumers. Particularly, the consumers request CES service in the amounts of total local  renewable generation they want to shift to supply their demand  over the contracted time period. We propose a quadratic price model for the ESO  charging  its consumers by their requested RUS.  
The underlying motivation of a quadratic price model is to capture the increasing marginal ES capital investment undertaken by the ESO to achieve the increasing RUS requested by the consumers.
Our CES model  adopts a general bargaining framework where  the ESO is allowed  to \emph{accept} or \emph{reject} the required RUS and the consumers exclusively determine their requested RUS to maximize their economic benefits.  
We make the following main contributions in this paper: 

\begin{itemize}
	\item[\emph{(C1)}] We propose a CES model  by suggesting the ESO to gain profit by
	selling renewable utilization service (RUS) to consumers.  The CES model  can secure higher consumer incentives
	over  individual ES (IES) installations.
	
		\item[\emph{(C2)}]   We formulate the problem as a Stackelberg game with the ESO as the \emph{leader} and the consumers as the \emph{followers}. We prove the CES model admits an \emph{equilibrium}  that is accessible by solving a mixed-integer linear programming (MILP) problem.

		\item[\emph{(C3)}]  We demonstrate  the desirable ESO profitability and  consumer incentives through numeric studies.  Particularly, we show the CES model is superior to the existing ES sharing models~\cite{zhao2017pricing, zhao2019virtual} both in ESO profit and consumer incentives. 
					
	    \item[\emph{(C4)}]   We  prove the CES model outperforms IES model in \emph{social welfare} theoretically. Moreover, we demonstrate the CES model can achieve near-optimal \emph{ social welfare} via numeric studies.

\end{itemize}

The remainder of this paper is structured as follows. 
In Section II, we survey the existing ES sharing models.
In Section III, we introduce the CES model and the Stackelberg game formulation. 
In Section IV, we study the \emph{equilibrium} of  CES model and present  the main theoretical results. 
In Section V, we study the economic benefits of CES model via numeric studies.
In Section VI, we conclude this paper.

\section{Literature}

Based on the sharing paradigms,  the existing ES sharing models can be divided into three groups: 
\emph{i)} peer-to-peer ES sharing, \emph{ii)} community ES sharing,  and \emph{iii)} third-party based ES sharing.  They mainly differ  in the ownership of  ES resources and  interactions among the  participants. 

Peer-to-peer ES sharing  refers to multiple consumers  sharing their under-utilized private ES with each other (see \cite{zhong2019online, tushar2016energy} for examples).  
Community ES sharing characterizes multiple users  cooperatively investing and using a central ES (see \cite{chakraborty2018sharing, yao2016stochastic, yao2017privacy, zhu2019credit},  and the references therein).   Clearly, these two kinds of ES sharing models differ in the ownership of ES resources. They also differ in the communication and interaction among the consumers. 
For peer-to-peer ES sharing,  the amount of shared ES capacity as well as the price are generally bargained exclusively by the end-users,  whereas in community ES sharing models the ES resources are  generally managed by a central coordinator. 
In the literature, most of the existing peer-to-peer ES sharing models focused on sharing under-utilized physical capacity (see \cite{tushar2016energy}), whereas community ES sharing models mostly studied the cooperation of multiple customers regarding the ES operation (see \cite{chakraborty2018sharing}).  
Therefore, the latter generally  favors the \emph{social welfare} but requires the participants to form a federation beforehand and adhere to centralized coordination. 



The last main  category is  third-party based ES sharing models where
an ESO  invests in ES resources  to gain profit by either leasing ES capacity (see \cite{zhao2017pricing, zhao2019virtual} for examples) or providing storage service to its consumers (see \cite{liu2017decision, liu2019research}  for examples). Obviously, this kind of ES sharing models involves two kinds of agents  with asymmetric roles (i.e., ESO and consumers), which significantly differs  from the other two categories with only peer participants. Therefore, this problem generally corresponds to \emph{market mechanism} design  to allocate  ES  economic benefits among the ESO and its consumers.

Different practices call for different sharing paradigms as discussed in \cite{vespermann2020access}.
For example, peer-to-peer sharing  can benefit  customers with private or isolated ES resources, while  community ES models allow  consumers to harness the maximal ES economic benefits  through full cooperation  from the planning stage of purchase.  Third-party based ES sharing models  can unload the capital cost barrier from the consumers and accelerate the ES deployment.  From the perspective of market structures,  peer-to-peer and third-party based ES sharing models can create a flexible market to allow the consumers to dynamically join or drop out, while community sharing models require a more-or-less fixed membership.   
Despite the flexibility and prospects, ESO profitability and desirable consumer incentives for third-party based ES sharing models are essential for exercising~\cite{dvorkin2016ensuring},  which  hasn’t been well addressed  and thus motivates this work.

\vspace{-3mm}
\section{The Problem}
In this section, we  first introduce our CES model  settings and then present the Stackelberg game formulation.

\subsection{CES model}
 
 \begin{figure}[h]
 	\centering
 	\includegraphics[width=2.8 in ]{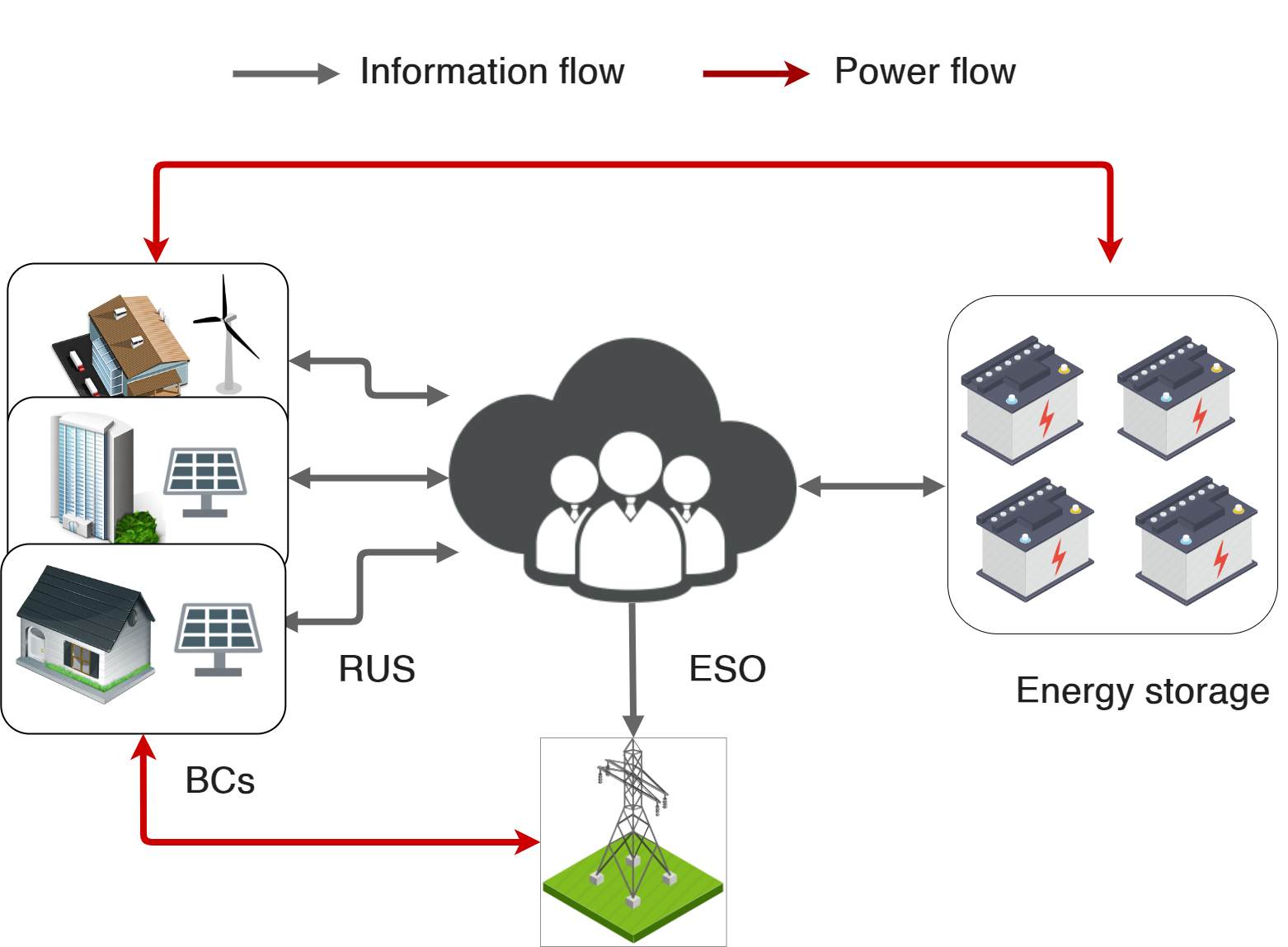}\\
 	\caption{The platform of CES model.}\label{Fig:CEMS}
 \end{figure}

The platform of a CES model in which the ESO provides ES service to its consumers is illustrated in Fig. \ref{Fig:CEMS}.  The platform is composed of three main blocks: \emph{building consumers} (BCs) with local renewable generation (e.g., solar and wind power),   the \emph{ESO} and \emph{ES facilities} owned by the ESO.  
The BCs can use the local renewable generation or procure electricity from the grid to satisfy their demand. Conversely,  their surplus local renewable generation can be sold back to the grid or stored if with ES.  We consider the practice that the BCs can choose to install private or isolated ES or  subscribe to the ESO for renewable utilization service (RUS): \emph{the  total amount of  local renewable generation $r_i$ (in \si{\kilo\watt\hour}) they want to shift  to supply their non-elastic demand at some other time by the CES}.  
In other word, the BCs can choose to store their  surplus renewable generation in the CES  and then discharge it  to supply their future demand. In such settings,  the ESO can seek profit by  investing in \emph{ES facilities} and providing  RUS  (i.e., ES service) to the BCs at scale .   

\emph{Illustration of RUS}: Fig. \ref{fig:RUT_example} gives an illustrative example of  RUS $r_i$ for BC $i$.  We use the curve to indicate a net generation profile  (renewable generation minus demand)  of BC $i$.  
  Intuitively,  the instant net renewable generation (above the axis) can not supply the BC's future demand (below the axis) without storage. 
  However, by subscribing to the ESO,  the surplus renewable generation indicated by $S^{+}_1$ and $S^{-}_2$ can be stored in the CES  and then used to supply the subsequent demand indicated by $S^{-}_1$ and $S^{-}_2$.  
 For this example, we have the RUS $r_i = S^{+}_1+S^{+}_2=S^{-}_1+S^{-}_2$ for BC $i$, which characterizes  the total amounts of renewable generation shifted by the CES to supply BC $i$'s demand over the contracted time period (e.g., one day).  
\begin{figure}[h]
	\centering
	\includegraphics[width=2.6in, height= 1.2 in]{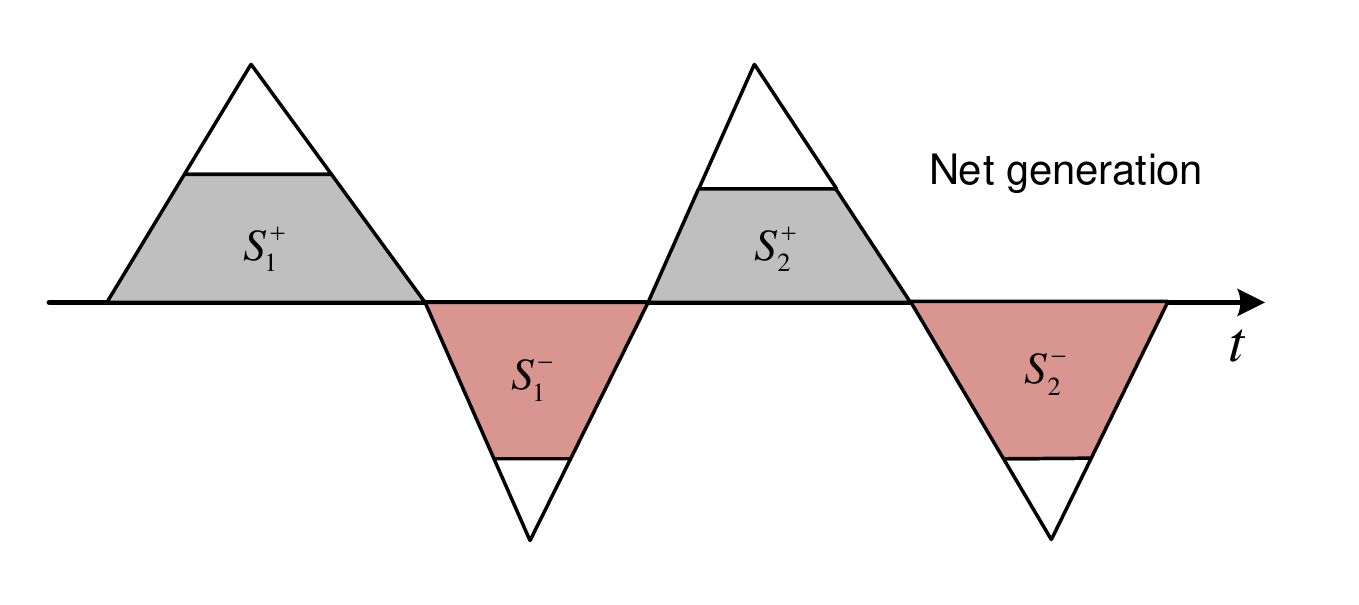}\\
	\vspace{-3mm}
	\caption{An illustrative example of RUS.}\label{fig:RUT_example}
\end{figure}

 The ESO decides the RUS price and the optimal ES (energy and power) capacity to invest in with the objective to maximize its profit. Particularly, the ESO should account for  the  consumer incentives, otherwise the BCs may  not join the business. 
   Based on the price, the BCs  determine their optimal RUS to request so as to minimize their total operation cost (i.e., electricity bill plus the RUS fee). For the ESO's sake, we allow  the ESO  to \emph{accept} or \emph{reject} the requested RUS  in the communication. 
Particularly,  our CES model adopts the similar idea of cloud computing where the consumers mostly care about accomplishing their computing tasks regardless of how the computing resources are dynamically distributed by the operator.  Analogously,  the CES model envisions the  consumers  only  care about the satisfaction of their requested RUS  instead of the ES operation behind.  In other words,
we assume the ESO will compute  the  optimal charging and discharging policies for the BCs provided with their predicted net generation profiles.

\vspace{-3 mm}
\subsection{Main assumptions}
We make the following main assumptions in this paper: 
\begin{itemize}
	\item[\emph{(A1)}] We consider the BCs buying electricity from  the grid at a fixed price that is much higher than that of selling back. 
    \item[\emph{(A2})] We do not consider energy sharing  among the BCs.
    \item [\emph{(A3)}] We focus on inelastic demand  and do not consider demand response of the BCs.
    \item[\emph{(A4)}] We study an ES sharing market with single ESO. 
\end{itemize}


\vspace{-5mm}
\subsection{Stackelberg game formulation}
We consider a group of  BCs $\mathcal{ N}: =\{1, 2, \cdots, N\}$ subscribing to the ESO for RUS over the  contract time period 
$\mathcal{T}: =\{ 1, 2, \cdots, T\}$.

Since both the ESO and  BCs are cost-aware,  the problems for them are both trade-offs.  For the BCs, if more local renewable generation is shifted to support demand, their  electricity bills will decrease (we assume zero marginal cost for renewable generation).  However, they would have to pay higher ES service fees to the ESO.  Conversely,  the ESO  can gain higher revenue by accepting more  RUS  but would have to invest  more  \emph{ES facilities}. 
This interaction between the ESO and  BCs  can be well captured by a Stackelberg game with the ESO as the \emph{leader} and the BCs as \emph{followers}.  
The ESO proposes an RUS price and then the BCs calculate their optimal RUS requests. 
Before we present the Stackelberg game formulation, we first address the RUS price model. 

  \emph{RUS price model}: Clearly, the RUS price model determines  the existence of an \emph{equilibrium} as well as  the \emph{efficiency} of ES resource (i.e., \emph{social welfare}).    In this paper, we adopt  a  quadratic price model for the ESO charging the BCs by their requested RUS: 
  \begin{equation} \label{eq:price_model}
  \begin{split}
   Q_i(r_i) = q_i ({r_i})^2
  \end{split} 
  \end{equation}
  where $q_i$ denotes the RUS  service price coefficient  and $r_i$ indicates the RUS of BC $i$. 
  $Q_i(r_i)$ captures the RUS payment of  BC $i$ for the requested RUS $r_i$. 

  There are two arguments for such  price model. 
   \emph{First},  it can  capture the increasing marginal ES investment made by the ESO to achieve the increasing  RUS for  the BCs. That is to say, 
   the ESO generally needs to spare increasing marginal ES resource  to satisfy the increasing  RUS. 
  This can be further interpreted that for a specific BC,  the  scheduling flexibility of its net generation profile decreases  with the increasing RUS,  thus more ES resource is required to achieve per-unit RUS.   
    This  characteristic can be captured by the quadratic price model \eqref{eq:price_model} as we have  $Q^{'}_i(r_i)= 2q_i r_i$, i.e., the marginal  ES service cost  is increasing w.r.t. $r_i$.
  \emph{Second}, we resort to a  discriminatory  price coefficient $q_i$ to account for the heterogeneous  net generation patterns of the BCs. 
  More specifically,  for two BCs requiring the same RUS, the ES resource occupied are usually different due to their different net generation patterns.  In the subsequent, we give the Stackelberg game for the CES model with the BCs as the \emph{followers} and the ESO as the \emph{leader}.

\textbf{BC $i$}: 
We consider a fixed grid price for the BCs and the purchase price $c^{\text{+}}_g$ from the grid is largely greater that of  the selling price $c^{\text{-}}_g$, i.e., $c^{\text{+}}_g  \gg c^{\text{-}}_g$.   
 This means  it is more sensible for the BCs to use local renewable generation for fulfilling their demand  instead of selling back to the grid.  Since the BCs are free to install private ES or join the CES model,  and they only care about the relative economic benefits of different options,  we use the BCs' cost  without ES (w/o ES) as  baseline and study the cost reductions for different ES models in this paper.  As for BC $i$ joining the CES model  and requesting RUS $r_i$, the cost reduction  comprises of two blocks:
  \emph{i)} the reduced electricity bill due to RUS  $\big( c^{\text{+}}_g-c^{\text{-}}_g \big) r_i $,  and \emph{ii)} the RUS payment to  the ESO  $Q_i(r_i) = q_i({r_i})^2$.  
Clearly,   BC $i$  has to make a trade-off between the  electricity bill reduction and  the  RUS  payment. 
 We use a collection of representative scenarios $\omega \in \Omega$ to capture the volatility of renewable generation for each BC.  Therefore, the problem for BC $i$  to decide  the optimal  RUS $r_i$   to minimize its total cost can be described as 
\begin{subequations}
	\begin{alignat}{4}
& J^{*}_{\sss \text{BC},i} = \min   -\big(c^{\text{+}}_g -c^{\text{-}}_g \big) r_i  + q_i({r_i})^2  \notag\\
&\label{eq:BC} \text{subject~to:} \tag{${\cal P}_{\sss \text{BC}, i}$}\\
&\label{eq:2a} \quad \quad r^{\min}_i \leq r_i \leq  r^{\max}_i.  \\
& \quad \quad \text{var:}~r_i.  \notag
	\end{alignat}
\end{subequations}
where $r^{\min}_i$ and $r^{\max}_i$ represent the minimum and maximum RUS  that can be requested by BC $i$.
$r^{\min}_i$  and  $r^{\max}_i$ are generally determined by the net generation profile of BC $i$.  For example, we usually have $r^{\min}_i=0$ and $r^{\max}_i$ as the amounts of  accumulated surplus renewable generation over the contracted time period $\mathcal{T}$. Essentially, each BC can not ask to shift more renewable generation than it produces. Therefore,  we can define $r^{\max}_i$ as 
\[r^{\max}_i =  \sum_{T \in \mathcal{T}}  \sum_{\omega \in \Omega} p_\omega [p^{\text{r}, \omega}_{i,t}-p^{\text{d}, \omega}_{i,t}]_{+} \]
where  $p^{\text{r}, \omega}_{i,t}$  and $p^{\text{d}, \omega}_{i,t}$  represent the renewable generation and inelastic demand of BC $i$ at time $t$, and $[x]_{+} = \max(x, 0)$.  

It is important to note that the BCs will only  participate in  the CES  if higher economic benefits are provided over the IES model. Specifically, define $J^{\sss \text{Ind}}_{\sss \text{BC}, i}$ as  the minimal total cost of  BC $i$ with IES model, which is available from \emph{off-line} computation, the necessary condition for BC $i$'s participation  can be captured by


\begin{equation}
J^{*}_{\sss \text{BC}, i} \leq J^{\sss \text{Ind}}_{\sss \text{BC}, i}
\end{equation}

\textbf{ESO}: The ESO seeks to gain profit by selling RUS to the BCs. To explore the  maximum economic benefits, we study the problem from the planning stage  by  considering the ES sizing.   In such  setting, the profit of ESO also comprises of two parts: \emph{i)} the ES capital investment, and \emph{ii)} the RUS payment  charged from the BCs. We denote the ES  energy and power capacity as $E$ (\si{\kilo\watt\hour}) and $P$ (\si{\kilo\watt}).
Considering the computation burden, we project the ES planning problem on a daily basis and use an amortized price model
 for energy capacity  $C_{\sss E}$ (s\$\si{\per{\kilo\watt\hour}}) and power capacity $C_{\sss P}$  (s\$\si{\per{\kilo\watt}}) for planning, which are obtained according to the projected 
 ES price 
 $100$\EUR\si{\per{\kilo\watt\hour}} and $300$\EUR\si{\per{\kilo\watt}} by 2025 \cite{pandvzic2018optimal}.   The ESO will determine the RUS price  and the ES size based on the requested RUS to maximize its profit. 
Besides, the ESO is allowed to \emph{accept} or \emph{reject} the requested RUS for profitability.  
Particularly, the ESO is authorized to determine the optimal charging and discharging policies for the BCs. 
Therefore, the problem for the ESO can be formulated  as 
\begin{subequations}
	\begin{alignat}{4}
 &J^{*}_{\sss \text{ESO} } =  \max~ \sum_{i \in \cal N}  q_i   ({r_i})^2 - ( C_{\sss E} E + C_{\sss P} P )  \notag\\
 &\label{eq:ESO}  \text{subject~to:}  ~~\eqref{eq:ChargingDischarging}, \eqref{eq:grid_trading}, \eqref{eq:ES_limits},   \eqref{eq:RUT_service}.  \tag{${\cal P}_{\sss \text{ESO}}$}\\
&  \text{var:} ~~E, P,   q_i,   y^{\omega}_i, ~\forall i \in \mathcal{N}, \omega \in \Omega. \notag
	\end{alignat}
\end{subequations}
\begin{itemize}

\item Constraints \eqref{eq:ChargingDischarging}  correspond to the ES operation policies for the BCs  defined as below. 
Constraints \eqref{eq:3a}-\eqref{eq:3b} impose the charging and discharging power limits $p^{\text{ch}, \max}$ and $p^{\text{dis}, \max}$ on the BCs. 
Constraint \eqref{eq:3c} tracks the stored energy for each BC with $\eta^{\text{ch}}, \eta^{\text{dis}} < 1$ denoting the ES roundtrip efficiency.  Particularly, the stored energy for different BCs are metered separately and we do not consider energy trading in this paper.
Constraint \eqref{eq:3d} prevents the BCs from over discharging their stored energy in the CES. 
The binary variables $x^{\text{ch}, \omega}_{i, t}, x^{\text{dis}, \omega}_{i, t}$ are introduced to impose the physical restrictions of simultaneous  charging and discharging for each BC. 
Constraint \eqref{eq:3f} indicates the BCs can only charge the CES  with  renewable energy  not  the procured energy from the grid.  This is because we  design to use the CES to shift  renewable generation. 
\begin{subequations} \label{eq:ChargingDischarging}
	\begin{alignat}{4}
& \label{eq:3a} 0 \leq p^{\text{ch}, \omega}_{i, t}  \leq x^{\text{ch}, \omega}_{i, t}  p^{\text{ch}, \max},  \\
&\label{eq:3b}  0 \leq  p^{\text{dis}, \omega}_{i, t}  \leq x^{\text{dis}, \omega}_{i, t}  p^{\text{dis}, \max},  \\
& \label{eq:3c}  e^{b, \omega}_{i, t+1}=e^{b, \omega}_{i, t}+p^{\text{ch}, \omega}_{i, t} \eta^{\text{ch}}- p^{\text{dis}, \omega}_{i, t} / \eta^{\text{dis}}, \\
&  \label{eq:3d}  e^{b, \omega}_{i, t}  \geq 0, \\
&  \label{eq:3e}  x^{\text{ch}, \omega}_{i, t} + x^{\text{dis}, \omega}_{i, t} \leq 1, x^{\text{ch}, \omega}_{i, t}, x^{\text{dis}, \omega}_{i, t} \in \{0, 1\},\\
& \label{eq:3f} p^{\text{ch}, \omega}_{i, t} \leq p^{r, \omega}_{i, t},~\forall \omega \in \Omega, \forall i \in \mathcal{ N}, t \in \mathcal{ T}.  
	\end{alignat}
\end{subequations}

\item Constraints \eqref{eq:grid_trading} model the balance of  BCs'  instantaneous supply and demand. 
At each instant $t$,  the BCs'  net  demand (i.e., $p^{\text{d}, \omega}_{i, t}>p^{\text{r}, \omega}_{i, t}$) is satisfied by the  energy discharged from the  CES   plus the grid purchase, or the net renewable generation (i.e., $p^{\text{d}, \omega}_{i, t}<p^{\text{r}, \omega}_{i, t}$)  is balanced by the energy charged into the CES and the  grid injections. In particular, we use the binary variables to indicate  $s_i$ whether the requested RUS of BC $i$ is  \emph{accept} ($s_1 =1$) or \emph{reject} ($s_i=0$)  by the ESO. Particularly, if BC $i$ (i.e., $s_i = 0$) is \emph{rejected},  the energy balance  constraint \eqref{eq:grid_trading} for BC $i$ will be relaxed.   Certainly,  to achieve energy balance,  the energy trading with the grid  should comply with the physical limits as imposed by constraint  \eqref{eq:4c}.  
\begin{subequations} \label{eq:grid_trading}
	\begin{alignat}{4}
& \label{eq:4a} p^{ \text{g+}, \omega}_{i, t}- p^{\text{g-}, \omega}_{i, t} \geq p^{\text{ch}, \omega}_{i, t}- p^{\text{dis}, \omega}_{i, t} + s_i p^{\text{d}, \omega}_{i, t}-p^{\text{r}, \omega}_{i, t},\\
&\label{eq:4b}  x^{\text{g+}, \omega}_{i, t} + x^{\text{g-}, \omega}_{i, t}  \leq 1,  ~x^{\text{g+}, \omega}_{i, t}, x^{\text{g-}, \omega}_{i, t} \in \{0, 1\}, \\
& s_i \in \{0, 1\}, \\
& \label{eq:4c} p^{ \text{g+}, \omega}_{i, t} \leq x^{\text{g+}, \omega} _{i, t} p^{g, \max}, ~ p^{\text{g-}, \omega}_{i, t} \leq x^{\text{g-}, \omega}_{i, t} p^{g, \max}, \\
& \quad \quad ~\forall \omega \in \Omega, \forall i \in \mathcal{ N}, t \in \mathcal{ T}. \notag 
	\end{alignat}
\end{subequations}

 \item Constraints  \eqref{eq:ES_limits} model the ES capacity limits. 
 Intuitively, the total amount of stored energy for all BCs can not exceed the ES energy capacity. 
 Analogously,  the net charging or discharging power must respect the  ES power capacity limits \eqref{eq:5a}.
\begin{subequations} \label{eq:ES_limits}
	\begin{align}
	\label{eq:5a} 0  \leq &\sum_{i \in \mathcal{N} }  e^{b, \omega}_{i, t}\leq E, \\
	\label{eq:5b}    -P  \leq  &\sum_{i \in \mathcal{N} } p^{ \text{ch}, \omega}_{i, t}-\sum_{i \in \mathcal{N} } p^{ \text{dis}, \omega}_{i, t}\leq P,  \\
   &  \quad \quad \quad ~\forall \omega \in \Omega, ~ t \in \mathcal{ T} \notag
	\end{align}
\end{subequations}

\item Constraint \eqref{eq:RUT_service} captures the satisfaction of  RUS:
\begin{equation} \label{eq:RUT_service}
\begin{split}
C_i(\bm{y}^{\omega}_i) \leq \overline{C}_i - ( c^{ \text{g+} }-c^{ \text{g-}} ) r_i, ~\forall i \in \mathcal{ N}
\end{split}
\end{equation}
where for brevity we use  $\bm{y}^{\omega}_i =[p^{\text{ch}, \omega}_{i, t} , p^{\text{dis}, \omega}_{i, t},   x^{\text{ch}, \omega}_{i, t}, x^{\text{dis}, \omega}_{i, t},  $ $e^{b, \omega}_{i, t}, p^{ \text{g+}, \omega}_{i, t},  p^{\text{g-}, \omega}_{i, t}, x^{ \text{g+}, \omega}_{i, t}, x^{\text{g-}, \omega}_{i, t}, s_i], \forall t, \omega$  to denote the concatenated decision variables associated with BC $i$. 
$\overline{C}_i$ denotes the electricity bill of BC $i$ with  w/o ES model, which can be obtained \emph{off-line}. 
  $C_i(\bm{y}^{\omega}_i)$ quantifies the incurred electricity bill of BC $i$ with the requested RUS $r_i$ defined as 
\[C_i(\bm{y}^{\omega}_i) = \sum_{\omega \in \Omega} p_\omega  \big( c^{\text{+}}_g \sum_{t\in \mathcal{T}} p^{\text{g+}, \omega}_{i,t} -c^{\text{-}}_g \sum_{t\in \mathcal{T}} p^{\text{g-}, \omega}_{i,t}\big)\]
The intuitive interpretation of constraint \eqref{eq:RUT_service} is that shifting at least $r_i$ units of local renewable generation for BC $i$  is equivalent to reducing its electricity bill by at least $( c^{ \text{+}}_g -c^{ \text{-}}_g ) r_i$. 

\end{itemize}

\section{Equilibrium and Main Results}

In this section, we study the existence  of the \emph{equilibrium} and  \emph{social welfare} of the CES model. 

\subsection{Obtain the equilibrium}
We first study the existence of \emph{equilibrium}.  For the Stackelberg game  formulation in Section III-B, we can obtain an explicit formula for the \emph{followers'} problem \eqref{eq:BC}  given the RUS price parameters $q_i$:
\begin{equation}
r_i =\Big[ \frac{c^{\text{+}}_g - c^{\text{-}}_g}{2q_i} \Big]_{r^{\min}_i}^{r^{\max}_i} 
\end{equation}
where $[x]^y_z = \max\{ \min\{ x, y\} \}$ indicates projecting $x$ into the segment $[y, z]$.

We assume the \emph{equilibrium} exists and denote it by $(q^{*}_i, r^{*}_i)$,  therefore we must have
\begin{equation} \label{eq:price}
\begin{split}
r^{*}_i  =\left\{
\begin{array}{ll}
r^{\min}_i, & q^{*}_i  > \frac{c^{\text{+}}_g - c^{\text{-}}_g}{2r^{\min}_i}  \\
\frac{c^{\text{+}}_g - c^{\text{-}}_g}{2q^{*}_i},  &   \frac{c^{\text{+}}_g - c^{\text{-}}_g}{2r^{\max}_i}  \leq q^{*}_i \leq  \frac{c^{\text{+}}_g - c^{\text{-}}_g}{2r^{\min}_i}  \\
r^{\max}_i, & q^{*}_i  < \frac{c^{\text{+}}_g - c^{\text{-}}_g}{2r^{\max}_i}  
\end{array}
\right.
\end{split}
\end{equation} 

We can derive from \eqref{eq:price} that
\begin{equation} \label{eq:price2}
 q^{*}_i =\frac{c^{\text{+}}_g - c^{\text{-}}_g}{2r^{*}_i}, ~ r^{\min}_i \leq r^{*}_i \leq r^{\max}_i
 \end{equation}


To obtain the \emph{equilibrium}, we can substitute \eqref{eq:price2} into problem \eqref{eq:ESO} and get the blended mixed-integer linear programming (MILP) problem: 
\begin{subequations}
	\begin{alignat}{4}
	&J^{*}_{\sss \text{ESO} } =  \max~ \sum_{i \in \cal N}  \frac{c^{\text{+}}_g - c^{\text{-}}_g}{2}  {r_i}- ( C_{\sss E} E + C_{\sss P} P )  \notag\\
	&\label{eq:ESO_reformulate}  \text{subject~to:}  ~~\eqref{eq:ChargingDischarging}, \eqref{eq:grid_trading}, \eqref{eq:ES_limits},   \eqref{eq:RUT_service}.  \tag{${\cal P}$}\\
	& \quad \quad  \quad \quad  r^{\min}_i \leq r_i \leq r^{\max}_i, \notag\\
	& \label{eq:9a}  \quad \quad  \quad \quad \frac{c^{\text{+}}_g - c^{\text{-}}_g}{2}  {r_i}  + C_i(\bm{y}^{\omega}_i)  \leq s_i J^{\sss \text{Ind}}_{\sss \text{BC}~i}, \\
	&\label{eq:9b} \quad \quad  \quad \quad r_i \leq s_i r^{\max}_i, ~\forall i  \in \mathcal{ N}.  \\
	&  \text{var:} ~~E, P, r_i,   y^{\omega}_i, s_i, ~\forall i \in \mathcal{N}, \omega \in \Omega. \notag
	\end{alignat}
\end{subequations}
where we translate  constraints  \eqref{eq:RUT_service} into constraint \eqref{eq:9a} to capture consumer incentives.
Constraint \eqref{eq:9b} models the \emph{accept} and \emph{reject} mechanism for  the ESO. Particularly, we have $r_i = 0$ if BC $i$ is rejected ($s_i=0$).

Therefore, the existence of \emph{equilibrium} corresponds to the solution of problem \eqref{eq:ESO_reformulate} and we have the  main results. 
\begin{theorem}
The Stackelberg game  for the CES model
admits an equilibrium. 
\end{theorem}

\begin{proof}
	
	To prove the existence of the \emph{equilibrium}, it suffices to prove at least one   optimal solution exists for problem \eqref{eq:ESO_reformulate}, which can be illustrated by two steps.
\emph{First},  we note  problem \eqref{eq:ESO_reformulate} is well-defined and  at least one feasible solution exists. 
\emph{Second}, problem  \eqref{eq:ESO_reformulate}  is compact as we have  $r^{\min}_i \leq r_i \leq r^{\max}_i$. 
Thus,  at least one optimal solution 
$\bm{r}^{*} = [r^{*}_i] ~\forall i $ and  $\bm{q}^{*} =[ \frac{c^{\text{+}}_g - c^{\text{-}}_g }{2r^{*}_i}] ~\forall i$ exists. This induces the existence of the \emph{equilibrium} for the CES model.

\end{proof}

\subsection{Social welfare}
In this subsection, we study  the  \emph{social welfare} of the CES model and compare it with the IES and CMES model.  As discussed in the literature,  community ES (CMES) models refer to multiple consumers cooperatively invest and share a central ES  which is managed by a central coordinator to maximize the community-wise economic benefits \cite{yang2020optimal},  and thus can be used to capture the \emph{social optima} of ES sharing. As for the social performance of the CES model, we have the following main results. 

\begin{theorem}    \label{thm:social_cost}
	The  social welfare of CES model is bounded by the  individual ES (IES) model and   community ES (CMES) model, i.e.,
	\[ \text{SC(CMES)} \leq  \text{SC(CES)} \leq  \text{SC(IES)}  \]
	where $SC(\cdot)$  indicates  the social cost of a specific ES model. 
\end{theorem}

\begin{proof}
	
	We structure our proof by the left-hand side (LHS) and the right-hand side (RHS), respectively.

	\emph{i)} We first prove the LHS, i.e., $\text{SC(CMES)} \leq  \text{SC(CES)}$. 
For the CES model, the \emph{social welfare} is characterized by the total cost: the  electricity bill plus the ES capital cost, i.e., 
	\begin{equation*}
	\begin{split}
	\text{SC(CES)} = \sum_{i \in \mathcal{ N}} \big\{ \overline{C}_i - ( c^{\text{+}}_g - c^{\text{-}}_g) r^{*}_i \big\} + C_{\sss E} E + C_{\sss P} P 
	\end{split}
	\end{equation*}
	where $r^{*}_i$ is the optimal RUS for  BC $i$  corresponding to the \emph{equilibrium} as discussed in Section IV-A. 
		
	For the CMES model \cite{yang2020optimal}, the \emph{social welfare} can be obtained by solving the following optimization problem: 
	\begin{subequations}
		\begin{alignat}{4}
		&	\text{SC(CMES)} = \min \sum_{i \in \mathcal{ N}} \big\{ \overline{C}_i - ( c^{\text{+}}_g - c^{\text{-}}_g) r_i \big\} + C_{\sss E} E + C_{\sss P} P  \notag\\
		&\label{eq:CMES}   \text{subject~to:}  ~~\eqref{eq:ChargingDischarging}, \eqref{eq:grid_trading}, \eqref{eq:ES_limits},   \eqref{eq:RUT_service}.  \tag{${\cal P}_{\sss \text{CMES}}$}\\
		& \quad \quad  \quad \quad  r^{\min}_i \leq r_i \leq r^{\max}_i, ~\forall i  \in \mathcal{ N}. \notag\\
		& \quad \quad  \quad \quad  \text{var:} ~~E, P, r_i,   y^{\omega}_i, ~\forall i \in \mathcal{N}, \omega \in \Omega. \notag
		\end{alignat}
	\end{subequations}

By comparing problem \eqref{eq:CMES}  and problem  \eqref{eq:ESO_reformulate}, we see the latter includes all the constraints of the former. Therefore,  all the feasible solutions of  problem \eqref{eq:ESO_reformulate} are  also feasible for problem  \eqref{eq:CMES}. 
Particularly, the optimal solution of problem   \eqref{eq:ESO_reformulate}, i.e., $r^{*}_i, \forall i$, must be  a feasible solution of  problem  \eqref{eq:CMES}.  Therefore, we have
\[ \text{SC(CMES)} \leq \text{SC(CES)}   \]

\emph{ii)} We prove the RHS, i.e., $ \text{SC(CES)} \leq \text{SC(IES)}$. 
Based on the problem definition, we know the BCs will join in the CES provided with the incentives  not less than IES model, otherwise they will choose to  install IES.   Therefore,  for  the proof, we only need to concentrate on the involved BCs in the CES denoted by  $\mathcal{S}: = \{ i \in \mathcal{ N} \vert s_i =1\}$.  
Intuitively, for the set of BCs $\mathcal{S}$,  we have the provided cost reduction by the ESO is  not less than the IES model, i.e.,
\begin{equation} \label{eq:rationality}
	\begin{split}
 - ( c^{\text{+}}_g - c^{\text{-}}_g) r^{*}_i  + q^{*}_i {r^{*}_i}^2  \leq& - ( c^{\text{+}}_g - c^{\text{-}}_g) \tilde{r}_i\\
&  +  C_{\sss E}\tilde{E}_i + C_{\sss P} \tilde{P}_i, ~\forall i \in \mathcal{S}. \\	\end{split}
\end{equation}
where $r^{*}_i$ and $q^{*}_i$ indicate the \emph{equilibrium}  of the CES model.  $\tilde{r}_i$, $\tilde{E}_i$ and $\tilde{P}_i$ denote the optimal RUS, ES energy and power capacity for the  IES model for BC $i$ with IES model.

Besides, we must have non-negative profit for the ESO, otherwise it would not  start the business, i.e.,
\begin{equation} \label{eq:profitable}
\begin{split}
C_{\sss E} E^{*} + C_{\sss P} P^{*} \leq \sum_{i \in \mathcal{S} } q^{*}_i ({r^{*}_i})^2 
\end{split}
\end{equation}

By combining \eqref{eq:rationality} and \eqref{eq:profitable}, we have
\begin{equation}
\begin{split}
& \sum_{i \in \mathcal{S}}   - ( c^{\text{+}}_g - c^{\text{-}}_g) r^{*}_i  + C_{\sss E} E^{*} + C_{\sss P} P^{*} \\
&\quad \quad  \leq \sum_{i \in \mathcal{S}} - ( c^{\text{+}}_g - c^{\text{-}}_g) \tilde{r}_i  +  C_{\sss E}\tilde{E}_i + C_{\sss P} \tilde{P}_i. 
\end{split}
\end{equation}

Equivalently, we have $ \text{SC(CES)} \leq \text{SC(IES)}$. 
\end{proof}           

\begin{remark}
\textbf{Theorem} \eqref{thm:social_cost} implies the CES model outperforms IES model in \emph{social welfare}.  In other words, the CES model can explore high economic benefits from  ES. However, the social performance of the CES model is upper bounded by the CMES model. This is reasonable  as the CMES model provides the social  optima through full cooperation. 
\end{remark}

\section{Numeric Results}          
In this section, we study the performance of the CES model via case studies.  
Particularly, we  empirically justify the rationality of the quadratic price model in Section V-B.  We study the ESO profitability and consumer incentives in Section V-C. Last but importantly, we study the \emph{social welfare} and ES efficiency in Section V-D. 

\subsection{Data and Parameters}                                                     
We set the case studies based  on real data for  building  demand \cite{BuildingDemand} and renewable generation (i.e., wind power and solar power) \cite{WindSolar}.  To account for the complementary features, we consider multiple types of buildings (i.e., office, hotel, school, hospital and restaurant) participating in the CES model. A typical demand curve and renewable generation (i.e., wind and solar power) profile  are shown in Fig. \ref{Fig: Demand_Renewable}(a) and  Fig. \ref{Fig: Demand_Renewable}(b), respectively. 
\begin{figure}[h]
	\centering
	\includegraphics[width=2.4 in, height= 1.4 in ]{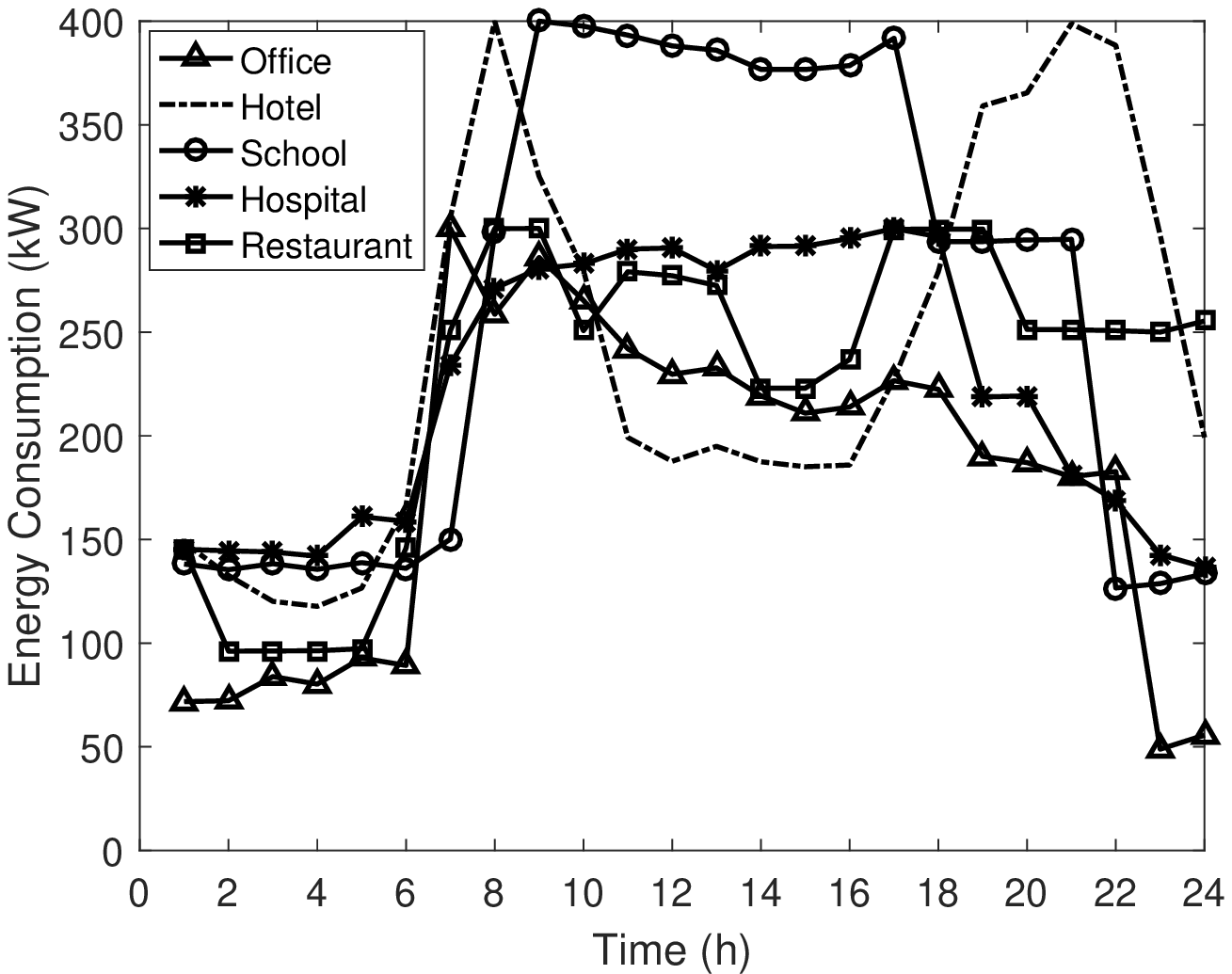}\\
	\includegraphics[width=2.4 in, height= 1.4 in ]{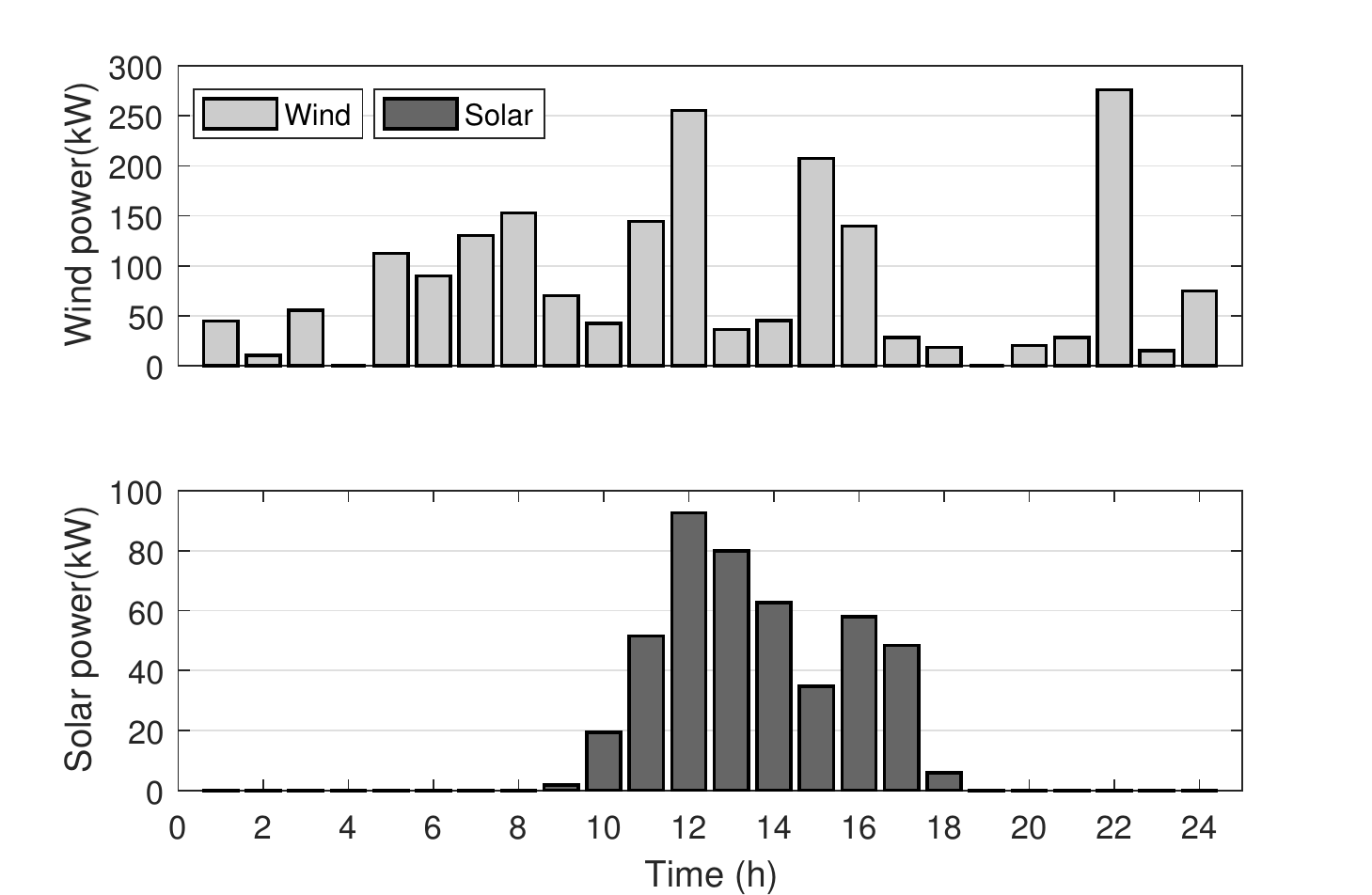}\\
	\caption{(a) A typical demand profile for different types of buildings (i.e., office, hotel, school, hospital, restaurant). (b) A typical wind and solar generation profile.}\label{Fig: Demand_Renewable}
\end{figure}


We set the  ES roundtrip efficiency as $\eta^{\text{ch}}, \eta^{\text{dis}} = 0.9$.  For the amortized ES capital price, we set the annual interest rate   $r= 0.06$ and the ES lifetime $L=10$ years. $\vert \Omega \vert =10$ representative scenarios are used to capture the volatility of renewable generation.  
 The maximum trading power  with the grid is set as  $p^{g, \max} = 1000$ {\small \si{\kilo\watt}} for each building. 
The charging and discharging rate limits for each BC are set as $p^{\text{ch}, \max}, p^{\text{dis}, \max} = 500 \si{\kilo\watt}$.  


\subsection{RUS Price Model}

We first empirically  justify the quadratic price model, and then compare the RUS price with the projections of  IES model.

To justify the quadratic price model, we first study the  marginal ES capital investment w.r.t. the RUS for each BC with the IES model.  We consider $N=5$ BCs of different types and  set the grid price as $c^{\text{g+}} =0.3$s\$\si{\per{\kilo\watt\hour}} and $c^{\text{g-}}=0$.   With IES model,   the minimum ES capital cost  of BC $i$ to  achieve RUS  $r_i$ can be obtained by solving the following problem: 
\begin{subequations}
	\begin{alignat}{4}
	\label{eq:IES}{Q}^{\sss \text{IES}}_i & (r_i ) = \min C_{\sss E} E_i + C_{\sss P} P_i \tag{$\mathcal{P}_{\text{IES}}$}\\
	&\text{subject~to}:  \notag\\
	& \quad  0 \leq p^{\text{ch}, \omega}_{i, t}  \leq x^{\text{ch}}_{i, t}  p^{\text{ch}, \max},  \\
	&  \quad 0 \leq  p^{\text{dis}, \omega}_{i, t}  \leq x^{\text{dis}}_{i, t}  p^{\text{dis}, \max},  \\
	& \quad  e^{b, \omega}_{i, t+1}=e^{b, \omega}_{i, t}+p^{\text{ch}, \omega}_{i, t} \eta^{\text{ch}}- p^{\text{dis}, \omega}_{i, t} / \eta^{\text{dis}}, \\
	& \quad  0 \leq e^{b, \omega}_{i, t}  \leq E_i, \\
	& \quad  p^{\text{ch}, \omega}_{i, t}, p^{\text{dis}, \omega}_{i, t}  \leq P_i, \\
	&  \quad x^{\text{ch}}_{i, t} + x^{dis}_{i, t} \leq 1, x^{\text{ch}}_{i, t}, x^{\text{dis}}_{i, t} \in \{0, 1\},\\
	&  \quad p^{\text{ch}, \omega}_{i, t} \leq p^{r, \omega}_{i, t},\\
	& \quad p^{ \text{g+}, \omega}_{i, t}- p^{\text{g-}, \omega}_{i, t} =p^{\text{ch}, \omega}_{i, t}- p^{\text{dis}, \omega}_{i, t} +p^{\text{d}, \omega}_{i, t}-p^{\text{r}, \omega}_{i, t},\\
	& \quad x^{\text{g+}, \omega}_{i, t} + x^{\text{g-}, \omega}_{i, t}  \leq 1,  ~x^{\text{g+}, \omega}_{i, t}, x^{\text{g-}, \omega}_{i, t} \in \{0, 1\}, \\
	& \quad p^{ \text{g+}, \omega}_{i, t} \leq x^{\text{g+}, \omega} _{i, t} p^{g, \max}, ~ p^{\text{g-}, \omega}_{i, t} \leq x^{\text{g-}, \omega}_{i, t} p^{g, \max}, \\
	& \quad \quad \quad ~\forall \omega \in \Omega, \forall i \in \mathcal{ N}, t \in \mathcal{ T}.  \notag\\ 
	&\quad  C_i(\bm{y}^{\omega}_i) \leq \overline{C}_i - ( c^{ \text{g+} }-c^{ \text{g-}} ) r_i, \\
	&\text{var:}~ E_i, P_i, y^{\omega}_i, \forall \omega \in \Omega. \notag
	\end{alignat}
\end{subequations}
where problem \eqref{eq:IES} adopts the notations of Section III  and the RUS $r_i$ can be regarded as an  input.

For each BC (BC1-BC5),  we can simulate  their  minimum ES capital cost w.r.t. the RUS $i$ by increasing $r_i$ from $r^{\min}_i=0$  to $r^{\max}_i$ with an incremental of $10$\si{\kilo\watt\hour}.  
We display the simulated  results  in Fig. \ref{fig:PriceModel} (dotted curves).  For each BC, we can find  a  quadratic function $\hat{Q}^{\sss \text{IES}}_i(r_i)=\hat{q}_i {r_i}^2$ that  well fits the simulated curve  (solid curves).   This implies for the BCs with IES installations, the marginal ES capital cost w.r.t. the RUS  can be approximately captured by a quadratic function.  
This somehow justifies the quadratic price model used by the ESO to charge the BCs for their requested RUS.

\begin{figure}[h]
	\centering
	\includegraphics[width=3.5 in]{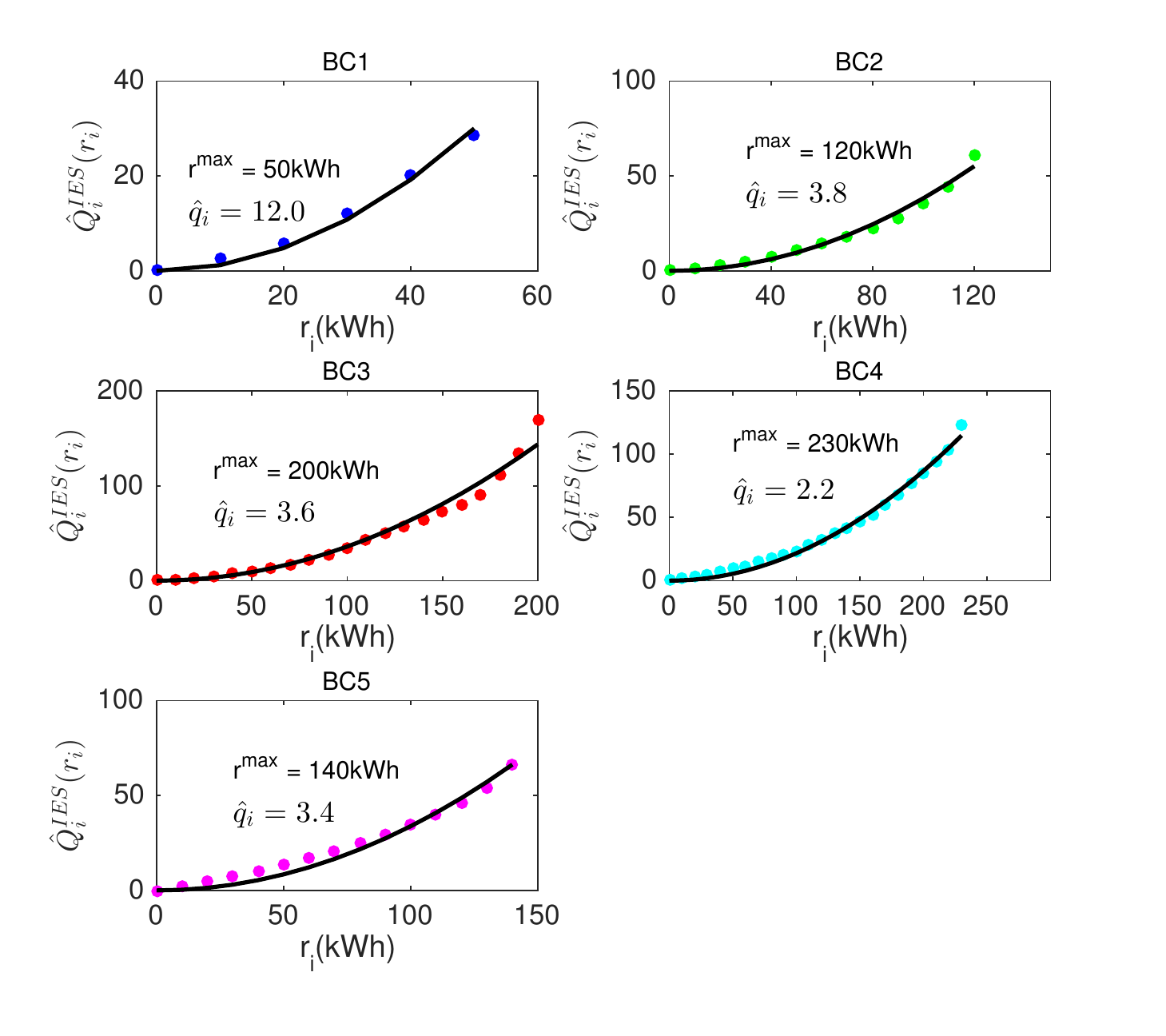}\\
	\vspace{-4mm}
	\caption{ES capital  cost w.r.t the RUS $r_i$ for BC1-BC5 with IES model (dotted curve: simulated results. solid curves: fitted results). }\label{fig:PriceModel}
\end{figure}

Subsequently, we compare the optimal RUS price of  CES model with the projections of IES model. 
For the  IES model,  the optimal RUS $\hat{r}_i$  and ES capital cost $\hat{Q}^{\sss \text{IES}}_i(\hat{r}_i)$ can be obtained from the simulations.  For comparison, we can obtain the projected RUS  price  $\hat{q}_i = \hat{Q}^{\sss \text{IES}}_i/{\hat{r}_i}^2$ for IES model. 
For the CES model, the optimal RUS $r^{*}_i$ and price $q^{*}_i$ are  obtained by solving problem \eqref{eq:ESO_reformulate}. 
For the two ES models,  the optimal RUS and (projected) price   are compared in TABLE \ref{tab:ES_cost_model}. 
We conclude the CES model provides a much lower  RUS price to the BCs compared with IES model(i.e, $q^{*}_i < \hat{q}_i$) and  enables higher renewable utilization (i.e., $r^{*}_i > \hat{r}_i$).  

\begin{table}[h]
		\setlength\tabcolsep{3pt}
	\centering
	\caption{Optimal RUS and (projected) RUS price  with the IES and CES model.}
	\label{tab:ES_cost_model}
	\begin{tabular}{ccccc}
		\toprule 
		 ES model & \multicolumn{2}{c}{\text{IES} }  &  \multicolumn{2}{c}{\text{CES}}   \\
		\#BC       &   $\hat{r}_i$({\scriptsize \si{\kilo\watt\hour}})  & $\hat{q}_i$({\scriptsize $\times 10^3$ })    & $r^{*}_i$({\scriptsize \si{\kilo\watt\hour}})     & $q^{*}_i$({\scriptsize $\times 10^3$ })    \\               
		\hline
		BC1       & 14.95    &11.4  &  50.81    & 2.95    \\
		BC2       & 45.97   & 3.5   &  90.74    & 1.65     \\
		BC3      & 44.61   &  3.3   &  86.35    &  1.74       \\
		BC4      & 81.45   & 2.3   & 114.88   & 1.31        \\
		BC5     & 45.63   & 3.5    &  98.42    & 1.52        \\
		\bottomrule 
	\end{tabular}
\end{table}

\subsection{ESO Profitability and Consumer Incentives}
In this part, we study the ESO profitability  and consumer incentives of the CES model. Particularly, we compare the CES model with  \emph{i)} an existing third-party based ES sharing  model, referred to  VES model  \cite{zhao2017pricing, zhao2019virtual},   and \emph{ii)} IES model.
For  the VES model,   the ESO charges  the BCs by their required energy  capacity (no charging and discharging rate limits) with  a linear price model. The price is determined by the ESO to maximize its profit.    
We investigate a number of case studies under the different scales (i.e., $N \in \{  5, 10, 15, 20\}$ and grid price settings (i.e.,  $c^{\text{+}}_g = \{0.25, 0.3, 0.35\}$s\$\si{\per{\kilo\watt\hour}} and $c^{\text{-}}_g=0$ for all cases). 
 Considering the  problem complexity, we depend on simulations to obtain the \emph{equilibrium} price for the VES model. Specifically,  we simulate the ES rental price  $q$ from $0.05$s\$\si{\per{\kilo\watt\hour}}  to $0.5$ s\$\si{\per{\kilo\watt\hour}} with an incremental of  $0.002$  s\$\si{\per{\kilo\watt\hour}}. It's clear that the \emph{equilibrium} price corresponds to  the point  with maximal ESO profit if exists.   
For the CES model, the \emph{equilibrium } price can be obtained by solving problem \eqref{eq:ESO_reformulate} based on CPLEX toolbox embedded in MATLAB.
We compare the ESO profit with the VES and CES model  under the different grid  price  settings  in TABLE \ref{tab:ESO_profit}. 
Firstly, we note the ESO profit  is increasing w.r.t the scale and grid price both with the VES and CES model. 
 While the former is straightforward, the latter  can be understood that the BCs would require higher RUS under a higher grid price, thus creating higher profit for the ESO.  
Notably, the CES model can generate much higher ESO profit (i.e., about 2-4 times) than the VES model  in the case studies.


\begin{table}[h!]
	\setlength\tabcolsep{3pt}
	\begin{center}
		\caption{ESO profit  with VES and CES model.}
		\label{tab:ESO_profit}
		\begin{tabular}{ccccccc}
			\toprule 
			\multirow{2}{*}{N}      & \multicolumn{2}{c}{{ $c^{\text{+}}_g= 0.25$s\$\si{\per{\kilo\watt\hour}} }}   &\multicolumn{2}{c}{{ $c^{\text{+}}_g= 0.3$s\$\si{\per{\kilo\watt\hour}} }}   &\multicolumn{2}{c}{{ $c^{\text{+}}_g= 0.35$s\$\si{\per{\kilo\watt\hour}}}}  \\
			&  VES    & CES            &   VES    & CES   &  VES    & CES   \\
			& (s\$)  &  (s\$)        & (s\$)   & (s\$)  & (s\$)  & (s\$) \\
			\hline
			5        & 3.22       & 23.83            &8.57           &  34.45        &   15.54            & 29.56\\
			10      & 25.74     & 73.83            & 47.96        &  100.09       &   67.54            & 102.83\\    
			15       &37.66     & 116.89          &70.41         &  157.17       &   91.79    & 179.34\\
			20    & 54.47    &191.23           &102.75       &  251.63       &  130.64     &  271.90 \\                   
			\bottomrule 
		\end{tabular}
	\end{center}
\end{table}

\begin{figure*}[h]
	\centering
	\includegraphics[width=6.4 in ]{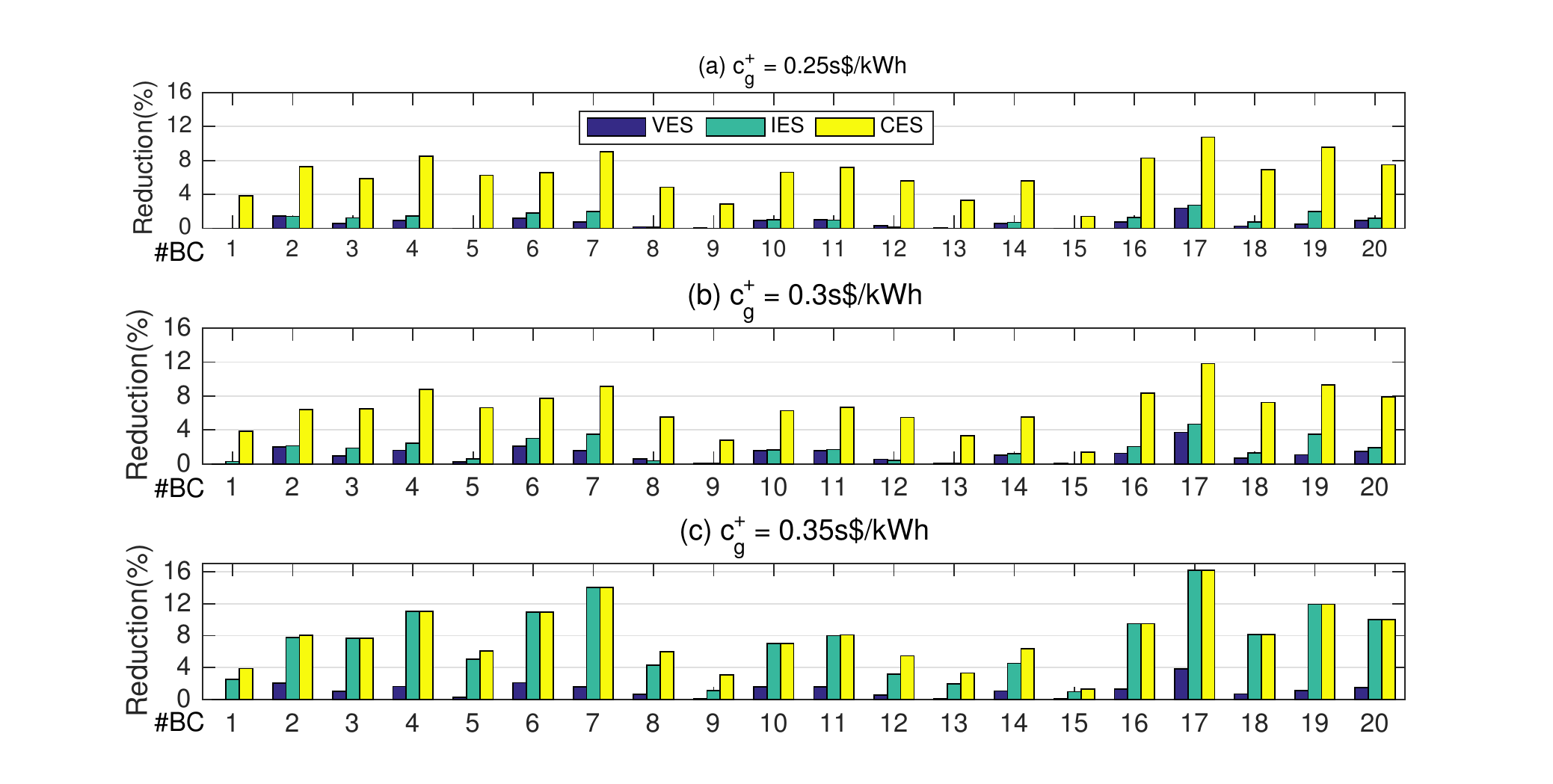}\\
	\caption{Cost reduction for BC1-BC20 with  different ES model (w/o ES model as \emph{benchmark}) under the grid price settings:  (a) $c^{\text{+}}_g = 0.25 s\$$\si{\per{\kilo\watt\hour}},  (b)  $c^{\text{+}}_g = 0.3 s\$$\si{\per{\kilo\watt\hour}}, (c) $c^{\text{+}}_g = 0.35 s\$$\si{\per{\kilo\watt\hour}}.}\label{fig:CostReductionN20}
\end{figure*}

 Further, we look into the consumer incentives with the different ES models.   For comparisons, we use each BC's electricity bill   without ES (w/o ES) as its \emph{baseline} and refer to  the cost reduction over the \emph{baseline} as consumer incentive. We compare the CES model with the IES and VES model.   For the IES model,  each BC installs individual ES to support local renewable utilization so as to minimize its total cost.  For display limits,  we use the case  with $N=20$  BCs as  an example and show the consumer incentives with the different ES models  in  Fig. \ref{fig:CostReductionN20} (we can observe the similar results for the other scales.   
First, we see the VES model provides the least cost reduction to the BCs  for all grid  price settings.  Therefore, we can envision the BCs would prefer to  install individual ES  rather than participate in the VES model considering the  economic gains alone.  In contrast,  the  CES model yields the most cost reduction for each BC over the other ES models. 
Further, we can observe some interesting phenomenon from   the  low to high  grid price settings.
For the IES model, the cost reduction for each  BC increases w.r.t. the grid price. 
Particularly, for the high grid price, i.e.,  $c^{\text{+}}_g = 0.35s\$\si{\kilo\watt\hour}$, the IES model can provide comparable  cost reduction to  most BCs (i.e.,  BC2-BC4, BC6-BC7, BC10-BC11, BC16-BC20) compared with  the CES model.
This implies the economic gap between the IES and CES model will diminish with the increase of grid price. This is reasonable that the economic barriers of IES model cased by the ES capital cost will drop when the grid price increase. We can also interpret the results  from the opposite that the difference over the different ES models will vanish if the ES capital cost is brought down to some boundaries. 


\subsection{Social welfare and ES efficiency}
We study the \emph{social welfare} and ES efficiency with the CES model under the different scale (i.e, $N =\{5, 10, 15, 20\}$) and grid price (i.e., $c^{\text{+}}_g = \{0.25, 0.3, 0.35\}$s\$\si{\per{\kilo\watt\hour}} and $c^{\text{-}}_g=0$) settings.
We compare the CES model with the  IES and VES model and   use 
 the w/o ES and CEMS model as  the \emph{references} of  lowest and highest \emph{social welfare}. 

For the  CES and VES  model, the \emph{social welfare} is characterized by  the total cost of all BCs minus the profit of ESO. For the IES and CMES model, the \emph{social welfare} equals to the total electricity bills of all BCs plus the ES capital cost. For the w/o ES model, the \emph{social welfare} defines the total electricity bills of all BCs as there is no ES installed. 

\begin{table}[h]
	\centering
	\caption{\emph{Social welfare} of different ES models.}
	 \label{tab:SocialCost}
	 \begin{tabular}{c}
	 	(a) grid price $c^{\text{+}}_g = 0.25 s\$$\si{\per{\kilo\watt\hour}}
	 \end{tabular}
	\begin{tabular}{cccccc}
		\toprule 
       \multirow{2}{*}{N}  &  \multicolumn{5}{c}{ES models} \\
       \cline{2-6}
                                        &    \text{w/o ES}  & \text{IES}  &   \text{VES} &   \text{CES} & \text{CMES} \\  
                                        & {\scriptsize ($\times 10^2$s\$)} &  {\scriptsize ($\times 10^2$s\$)} & {\scriptsize ($\times 10^2$s\$)} & {\scriptsize ($\times 10^2$s\$)}  & {\scriptsize ($\times 10^2$s\$)}  \\
        \hline
         5          &  10.73        &10.65        & 10.70        &10.02          & 9.90    \\
         10        &  23.62         &23.42       & 23.22        & 21.67         & 21.40   \\
        15         &  35.60         & 35.37      & 35.05        & 32.49         & 32.30   \\
        20         &  47.19          &46.77      & 46.36         &42.31         &42.08\\
		\bottomrule 
	\end{tabular}
~\\
~\\
~\\
	 \begin{tabular}{c}
	(b) grid price $c^{\text{+}}_g = 0.3 s\$$\si{\per{\kilo\watt\hour}}
\end{tabular}
	\begin{tabular}{cccccc}
	\toprule 
	\multirow{2}{*}{N}  &  \multicolumn{5}{c}{ES models} \\
	\cline{2-6}
	&    \text{w/o ES}  & \text{IES}  &   \text{VES} &   \text{CES} & \text{CMES} \\  
	& {\scriptsize ($\times 10^2$s\$)} &  {\scriptsize ($\times 10^2$s\$)} & {\scriptsize ($\times 10^2$s\$)} & {\scriptsize ($\times 10^2$s\$)}  & {\scriptsize ($\times 10^2$s\$)}  \\
	\hline
	5          &  12.88        & 12.70      &12.68          & 11.88    &11.78   \\
	10        & 28.34         & 27.92      &27.57          & 25.66    & 25.48   \\
	15        & 42.72         & 42.21      &41.64          & 38.66    & 38.53   \\
	20        & 56.63         & 55.74      &  54.99        &  50.46   & 50.15\\
	\bottomrule 
\end{tabular}
~\\
~\\
~\\	
	 \begin{tabular}{c}
	(b) grid price $c^{\text{+}}_g = 0.35 s\$$\si{\per{\kilo\watt\hour}}
\end{tabular}
\begin{tabular}{cccccc}
	\toprule 
	\multirow{2}{*}{N}  &  \multicolumn{5}{c}{ES models} \\
	\cline{2-6}
	&    \text{w/o ES}  & \text{IES}  &   \text{VES} &   \text{CES} & \text{CMES} \\  
	& {\scriptsize ($\times 10^2$s\$)} &  {\scriptsize ($\times 10^2$s\$)} & {\scriptsize ($\times 10^2$s\$)} & {\scriptsize ($\times 10^2$s\$)}  & {\scriptsize ($\times 10^2$s\$)}  \\
	\hline
	5          & 15.03        &14.71      & 14.73       & 13.64   &13.63    \\
	10        &33.06         & 32.34     &  31.93     & 29.54   & 29.54   \\
	15        & 49.84        & 48.95      & 48.32     & 44.76   & 44.73   \\
	20        & 66.07        & 64.56      &  64.03    & 58.23   & 58.19\\
	\bottomrule 
\end{tabular}
\end{table}

We report the \emph{social welfare}  of different cases  in TABLE \ref{tab:SocialCost}. 
First of all,  the  w/o ES and CMES model provide the  lowest and highest \emph{social welfare} for each case as anticipated.  Obviously, the CES model outperforms the IES and VES model in \emph{social welfare} as the former provides lower social cost over the other two for each case.  Besides, we see the VES and IES model are comparable in social cost  for each case. 

Further, to quantify  the social performance gaps of different ES models,  we define a relative social cost (RSC) metric that captures the normalized  social performance of different ES models over the CMES and w/o ES model, i.e., 
\[ \text{RSC}(x) = \frac{\text{SC}(x)-\text{SC}(\text{CMES})}{\text{SC}(\text{w/o ES}) - \text{SC}(\text{CMES})}\]
where  $x$ denotes an ES model, i.e., $x \in \{\text{IES},  \text{VES}, \text{CES}\}$.  
Intuitively, we have  $\text{RSC}(\text{CMES}) = 0 $ (social optima) and $\text{RSC}(\text{w/o ES}) = 1$. 
Besides, we note  the smaller the RSC, the higher of \emph{social welfare}.
For comparisons, we visualize the RSC metric of different ES models under the different scale  (i.e., $N \in \{  5, 10, 15, 20\}$)  and grid price (i.e., $c^{\text{+}}_g = \{0.25, 0.3, 0.35\}$s\$\si{\per{\kilo\watt\hour}})  settings 
in Fig. \ref{fig:relative social cost}. 
Firstly, we see the RSC of all  ES models (i.e., IES, VES, CES) are bounded by the w/o ES  and CMES model (i.e., $\text{RSC}(\text{IES}), \text{RSC}(\text{VES}), \text{RSC}(\text{CES}) \in [0, 1]$).
Besides, the IES and VES only achieve minor social superiority over the lower  baseline ( i.e., RSC(w/o ES) =0)  for all  cases.  
Nevertheless, the CES model dominates  the IES and VES model for all case in \emph{social welfare}  (i.e., $\text{RSC}(\text{CES}) \ll \text{RSC}(\text{IES}), \text{RSC}(\text{VES}$) ).
More notably,  we see the  CES model achieves the near \emph{social optima} (i.e., $\text{RSC}(\text{CES}) \approx 1$).

\begin{figure}[h]
	\centering
	\includegraphics[width=3.2 in]{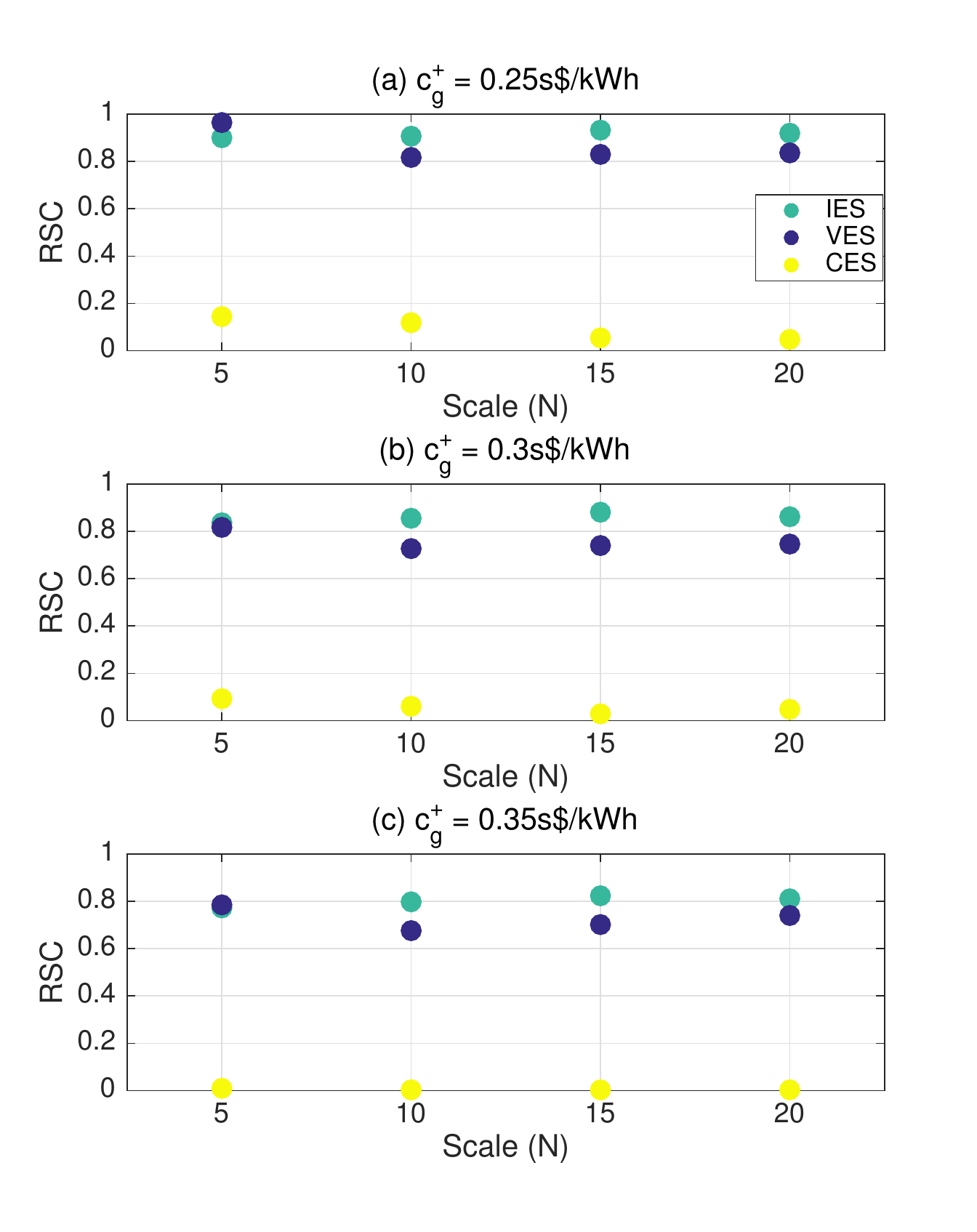}\\
	\caption{Relative social cost (RSC) of different models (w/o ES and CMES as \emph{benchmarks}: RSC(w/o ES) = 0 and RSC(CMES) = 1) under the grid price settings (a) $c^{\text{+}}_g = 0.25 s\$$\si{\per{\kilo\watt\hour}},  (b) $c^{\text{+}}_g = 0.3 s\$$\si{\per{\kilo\watt\hour}}. (c)  $c^{\text{+}}_g = 0.35 s\$$\si{\per{\kilo\watt\hour}}.}\label{fig:relative social cost}
\end{figure}

The high \emph{social welfare} generally corresponds to  high resource efficiency.  To further demonstrate it,  we study the ES efficiency with the CES model by comparing with the VES model. Specifically, we study the ES utilization   indicated by the State-of-Charge (SoC) and charging/discharging rates  over the contracted time period.  
We use the case with $N=10$ BCs and grid price $c^{\text{+}}_g = 0.35 s\$$\si{\per{\kilo\watt\hour}} for a randomly picked scenario as an example (similar results are observe in  the other cases).  
The  charging/discharging rates (\si{\kilo\watt}) and the SoC (\%)  for the BCs (BC1-BC10) with the two ES  models are contrasted  in Fig. \ref{fig:ChargingDischarging} and Fig. \ref{fig:SoCState}, respectively.
Notably, we see  CES model exhibits  much higher ES utilization than the VES model both for  energy capacity (i.e., SoC) and power capacity (i.e., charging and discharging rates).  Moreover, we can observe some interesting phenomenon regarding the ES operation patterns over the two ES models.  Specifically, from Fig. \ref{fig:ChargingDischarging},   we observe the BCs mostly charge or discharge synchronously and do not see much cancellations in sharing with the VES model as expected in \cite{zhao2017pricing, zhao2019virtual}.  This can be understood that though the ESO expects the BCs to charge and discharge in a complementary manner,  the linear price model based on energy capacity does not provide the coordination signals.   
In other words, the BCs  operate their contracted energy capacity separately, resulting in the low utilization of  ES resources, especially the power capacity.  Actually, this leads to  the short profitability  of  ESO as discussed in Section V-C.  
Nevertheless, we can observe quite different results for the CES model. At each instant, we can see both charging and discharging  from  different BCs, resulting in frequent cancellation and high ES utilization both for the energy and power capacity.  This  demonstrates the favorable performance of the CES model in shaping ES efficiency. 

 \begin{figure}[h]
	\centering
	\includegraphics[width=3.2 in]{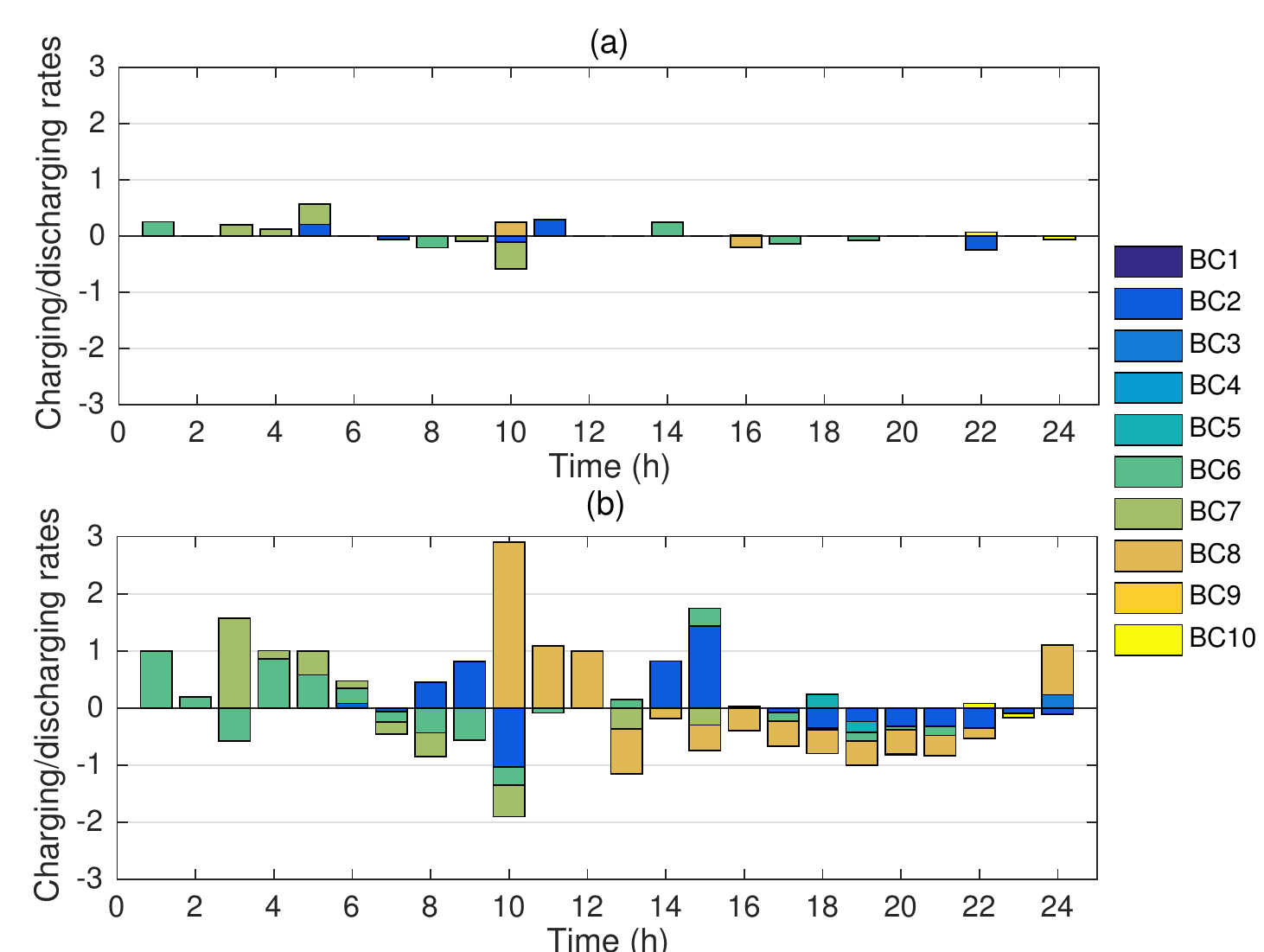}\\
	\caption{Charging/discharging rates of BC1-BC10 with (a) VES model,  and  (b) CES model.}\label{fig:ChargingDischarging}
\end{figure}

 \begin{figure}[h]
	\centering
	\includegraphics[width=3.2 in]{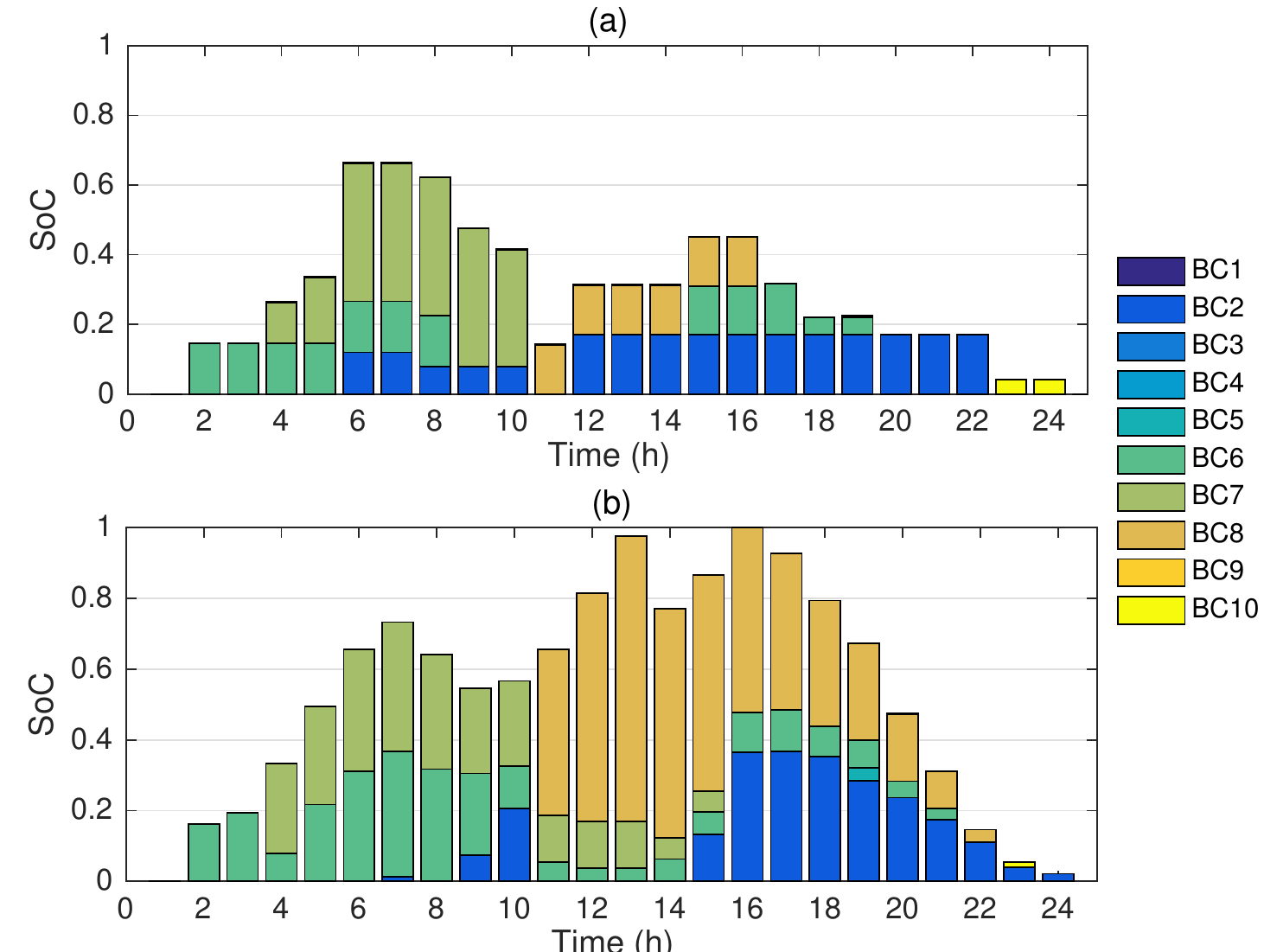}\\
	\caption{State-of-charge (SoC) of  BC1-BC10 with (a) VES model,  and (b) CES model.}\label{fig:SoCState}
\end{figure}

\section{Conclusion}
This paper proposed a cloud energy storage (CES) model for enabling distributed renewable integration of building consumers (BCs). Different from most existing ES sharing models that the entity gains profit by leasing energy or power capacity, our CES model is deployed by allowing the energy storage operator (ESO) to sell renewable utilization service (RUS) to its consumers. To capture the increasing marginal ``effort"  (i.e., ES resources) made by the ESO  to achieve the requested RUS, we adopted a quadratic price model which was  empirically justified through simulations.  We formulated the interactions between the ESO and the BCs as a Stackelberg game and demonstrated the existence of an \emph{equilibrium}. Besides, we showed the CES model can achieve better \emph{social welfare} than individual ES (IES) model.  The superior performance of the CES model over the IES model and an existing ES sharing model (referred to VES model) are also demonstrated via case studies. Specifically, we showed the CES model can provide 2-4 times ESO  profit over the VES model. Meanwhile, higher cost reduction  are provided to the BCs.  Last and noteworthy, we showed the CES model can achieve near \emph{social optima} and high ES efficiency (i.e., utilization) which are not provided by the other ES  sharing models.  Conclusively, this paper can work as an example to show how market design and sharing economy can shape the efficiency of energy systems.

\ifCLASSOPTIONcaptionsoff
  \newpage
\fi

\bibliographystyle{ieeetr}
\bibliography{references}

%
%
%

\end{document}